\documentclass[12pt, reqno,a4paper]{amsart}
\usepackage[left=2.5cm,top=2cm,right=2cm,bottom=2cm,bindingoffset=0.5cm]{geometry}

\usepackage{enumerate}
\usepackage{amssymb}
\usepackage{amsfonts}
\usepackage{amsmath,amscd,textcomp}
\usepackage{latexsym}
\usepackage{amsthm, amsfonts, graphicx, color} 
\usepackage{placeins}
\usepackage{longtable}
\usepackage{mathtools}                 
\usepackage[colorlinks=true, linkcolor=blue, citecolor=blue,urlcolor=blue,bookmarks=false,hypertexnames=true]{hyperref} 
\usepackage[dvipsnames]{xcolor}
\usepackage{mathrsfs}
\usepackage{float}
\usepackage{tabularx}
\usepackage{booktabs}
\usepackage{array}
\usepackage{multirow}
\usepackage{cleveref}

\definecolor{headerblue}{RGB}{0,102,204}
\definecolor{rowblue}{RGB}{230,240,255}
\newcolumntype{S}{>{\small}l}
\allowdisplaybreaks
\newtheorem{thm}{Theorem}[section]

\newtheorem{theorem}[thm]{Theorem}

\theoremstyle{definition}

\theoremstyle{remark}
\newtheorem{remark}[thm]{Remark}
\numberwithin{equation}{section}

\usepackage{tikz}
\usetikzlibrary{arrows.meta, positioning, shapes, shadings}
\usepackage{orcidlink} 

\DeclareMathSymbol{\subsetneqq}{\mathbin}{AMSb}{36}

\begin{document}

\begin{center}
{\bf\Large Dynamics of Unemployment with Discouraged Workers: A Nonlinear Mathematical Model}  \\
{\bf Poushali Das$^{a}$, Shraddha Sachan$^{a,\orcidlink{0009-0003-6118-099X}}$, Mukul Chauhan$^{a, \orcidlink{0009-0009-1332-3005}}$, Narendra Kumar$^{b, \orcidlink{0000-0003-0647-6600}}$,\\ Amit K. Verma$^{a,\orcidlink{0000-0001-8768-094X}}$}
\end{center}
\vspace{0.1in}
{\small
\textit{$^{a}$Department of Mathematics, Indian Institute of Technology Patna, Bihta, Patna 801103, (BR) India.} \\
\textit{$^{b}$Department of Mathematics, Indian Institute of Technology Jodhpur, IIT Jodhpur, Karwar, Rajasthan, India.}
}
\vspace{2mm}\\
E-mail: \texttt{poushalidas.acad@gmail.com};\\ \texttt{shraddhasachan8@gmail.com};\\ \texttt{mukul\_2421ma13@iitp.ac.in};\\
\texttt{narendrakumar.knp@gmail.com};\\ \texttt{akverma@iitp.ac.in}
\vspace{0.1cm}\\
$^{\star}$Corresponding author: \texttt{akverma@iitp.ac.in}

\parindent=0mm \vspace{.2in}
{\small {\bf Abstract} In this article, we formulate and analyze a new non-linear mathematical model to describe the dynamics of unemployment with a discouraged working population. We consider five dynamic variables, namely, unskilled unemployed individuals, skilled unemployed individuals, discouraged individuals, employed persons, and job vacancies. Furthermore, we determine the equilibrium points of the dynamic system and investigate their local stability. To demonstrate the results, we conduct numerical simulations by presenting solution trajectories and analyzing how variations in key parameters influence the states of the dynamical variables. 
	
	\parindent=0mm \vspace{.1in}
	{\bf{Keywords}} Mathematical model; Unemployment; Local stability; Employment; Routh-Hurwitz; Numerical simulation
	
	\parindent=0mm \vspace{.1in}
	{\bf {Mathematics Subject Classification:}} 97M20, 34A34, 34C28}
    \section{Introduction}\label{S1}

    \parindent=0mm \vspace{.1in}
     Unemployment is one of the most serious social and economic problems faced by countries across the world. It affects economic stability, government finances, and social welfare systems, particularly in economic structures that rely heavily on public support. In simple terms, unemployment refers to a situation where people who are able and willing to work cannot find suitable jobs. In recent decades, unemployment has increased due to rapid population growth, technological advances such as automation, artificial intelligence \cite{mutascu2021artificial}, and major health crises, including the COVID-19 pandemic \cite{blustein2020unemployment}. Apart from economic difficulties, unemployment also creates serious social problems. It can negatively affect mental health, increase stress and depression, and may even lead to higher levels of criminal behaviour \cite{kidwai2015psychological,murphy1999effect}.

   \parindent=8mm \vspace{.1in}
   Motivated by these concerns, considerable attention has been devoted to the development of mathematical models to understand the dynamics of unemployment and to evaluate the effectiveness of policy interventions. One of the early contributions in this direction was made by Nikolopoulos and Tzanetis \cite{nikolopoulos2003model}, who proposed a mathematical framework for housing allocation to homeless families following natural disasters. Building upon this idea, Misra and Singh \cite{misra2011mathematical} introduced a nonlinear unemployment model incorporating three dynamical variables: unemployed individuals, regularly employed individuals, and temporarily employed individuals. 
   
   \parindent=8mm \vspace{.1in}
   Subsequently, Misra et al. \cite{misra2013delay} extended this framework by incorporating time delays and demonstrated that while the creation of new vacancies can reduce unemployment, delays in vacancy generation may cause instability of the system.  Sirghi et al. \cite{sirghi2014dynamic} further expanded the model by adding two new variables that describe the number of jobs available in the market and the number of immigrants. Subsequent studies by Pathan and Bhathawala \cite{pathan2015mathematical, pathan2016unemployment} examined the roles of self-employment and examined job competition between unemployed individuals and new immigrants, while Munoli and Gani \cite{munoli2016optimal} introduced an optimal control approach to assess policy-driven unemployment reduction. More recently, Neamtu \cite{harding2018dynamic} developed a dynamic unemployment model incorporating migration and delayed policy intervention, formulated with five interacting state variables.

\parindent=8mm \vspace{.1in}
Subsequently, Munoli et al. \cite{munoli2017mathematical} introduced a more comprehensive model involving unemployed individuals, temporary and regular employment, and corresponding vacancies. Through optimal control analysis, they demonstrated that targeted government interventions can significantly mitigate unemployment. Similarly, Pathan and Bhathawala \cite{pathan2017mathematical} investigated a four-variable nonlinear model and concluded that prompt job creation by both public and private sectors plays a crucial role in reducing unemployment levels.

\parindent=8mm \vspace{.1in}
Galindro and Torres \cite{galindro2017simple} applied a mathematical unemployment model to the Portuguese labour market, incorporating unemployed individuals, employed individuals, and vacancies, and proposed an associated optimal control strategy. Pathan and Bhathawala \cite{pathan2017mathematical1} analysed job competition between native unemployed individuals and migrants, accounting for delayed and non-delayed vacancy creation as well as self-employment initiatives.

\parindent=8mm \vspace{.1in}
The relationship between unemployment and crime has been widely studied. Sundar et al. \cite{sundar2018does} developed a nonlinear model incorporating unemployment, employment, criminal activity, jail, and detention classes, demonstrating that increased employment opportunities reduce criminal behaviour. Maalwi et al. \cite{al2018unemployment} proposed a resource-constrained unemployment model tailored to economically weaker countries, demonstrating how employment-related parameters influence the reduction of unemployment.

\parindent=8mm \vspace{.1in}
More recent studies have focused on skill development and labour market efficiency. Singh et al. \cite{singh2020combating} investigated the impact of skill enhancement on unemployment using a five-variable model, concluding that increased skill development reduces unemployment while boosting employment. Sheikh et al. \cite{al2021mathematical1} and Misra et al. \cite{misra2022modeling} further analysed the effects of limited job opportunities and skill development, respectively, on unemployment dynamics.

\parindent=8mm \vspace{.1in}
Additional contributions include models incorporating recruitment delays \cite{verma2021getting,verma2023modeling}, training programmes \cite{al2021mathematical}, and industrialization \cite{singh2023modeling}. Ashi et al. \cite{ashi2022study} examined unemployment control strategies in the absence of intervention. In contrast, Mamo et al. \cite{mamo2024modelling} investigated the interplay between corruption, economic growth, and unemployment, emphasizing the importance of anti-corruption measures for economic stability.

\parindent=8mm \vspace{.1in}
Most recently, Misra et al. \cite{misra2025graph} developed a three-class unemployment model distinguishing between unskilled unemployed, skilled unemployed, and employed individuals, and extended their analysis using graph-theoretic techniques to capture employment diffusion through social networks. Their results highlighted the critical role of connectivity and reproduction thresholds in sustaining employment growth. In a related study, Ifeacho and Parra \cite{ifeacho2025mathematical} proposed a nonlinear model to examine the coupled dynamics of economic growth, corruption, unemployment, and inflation, with particular emphasis on the influence of inflation on unemployment.

\parindent=8mm \vspace{.1in}
To address the problem of unemployment and encourage people the Government of India is trying to train the youth through various programs like PMKVY (Pradhan Mantri Kaushal Vikas Yojana), Samarth (Scheme For Capacity Building In Textile Sector), CTS (Craftsmen Training Scheme), SANKALP (Skill Acquisition and Knowledge Awareness for Livelihood Promotion), NAPS (National Apprenticeship Promotion Scheme), JSS (Jan Shikshan Sansthan), PMKK (Pradhan Mantri Kaushal Kendra), PMVBRY (Pradhan Mantri Viksit Bharat Rojgar Yojana), NCS (National Career Service), Pradhan Mantri Ujjwala Yojana (PMUY), AVTS (Advanced Vocational Training Scheme), CITS (Crafts Instructor Training Scheme), Skill Loan Scheme, PMVK (PM Vishwakarma Yojana), TITP (Technical Intern Training Program), NSDM (National Skill Development Mission). All these programs help people who are in discouraged classes to revive themselves and do better in life, achieving happiness. 

\parindent=8mm \vspace{.1in}
 In this work, we have proposed a new non-linear mathematical model of unemployment with five dynamical variables: 
\begin{enumerate}[(i)]
    \item Set of unskilled unemployed people ($U_u$), 
    \item Set of skilled unemployed people ($U_s$), 
    \item People employed people ($E$), 
    \item Set of discouraged people ($D$), 
    \item No. of vacancies ($V$). 
\end{enumerate}
We have formulated the non-linear model, calculated the region of attraction, and obtained critical points of the system. Moreover, we have analysed the stability of the critical points by computing eigenvalues and using the Routh-Hurwitz criterion. Finally, we have performed a numerical simulation using the fourth-order Runge-Kutta method to validate our obtained results. In our view, this is a novel concept that introduces a class of ``discouraged people". This concept originates from individuals who experience dissatisfaction with their employment or overall life conditions, either because competitive pressures and structural constraints prevent them from achieving desired goals, or because population growth and labour market competition result in unemployment. Discouraged workers are those who are willing and able to work but have withdrawn from active job search due to persistent rejection, limited employment opportunities, inadequate skills, or heightened market competition. In this framework, discouragement denotes a condition in which an individual's level of well-being falls below their expected standard.

\parindent=8mm \vspace{.1in}
This article is organized as follows. Section~\ref{S3} formulates the model. Section~\ref{S4} derives the equilibrium points, and Section~\ref{S5} examines their stability. Numerical simulations illustrating the analytical results are presented in Section~\ref{S6}. The conclusions are given in Section~\ref{S7}.
\section{The Proposed Mathematical Model}\label{S3}
In the mathematical modeling process, we assume that the total number of unskilled unemployed persons is increasing at a constant rate $\Gamma$. Furthermore, we assume that all populations entering the system are at least 18 years of age. Additionally, we are considering the following assumptions:
\begin{enumerate}[(a)]

    \item The rate at which unskilled unemployed individuals develop skills is proportional to the total unskilled unemployed and the skilled unemployed populations.
    \item The rate at which unskilled unemployed persons move to the discouraged class is proportional to the unskilled unemployed and discouraged populations.
    \item The rate of movement of the skilled unemployed class to the discouraged class is very low, so we are neglecting that parameter.
\end{enumerate}
We describe the model in the following Fig.~\ref{fig:figure1}
\begin{figure}[H]
    \centering
    \begin{tikzpicture}[
    node distance=3.2cm,
    every node/.style={font=\small},
    comp/.style={
        ellipse,
        draw=black,
        thick,
        minimum width=1.8cm,
        minimum height=1.2cm,
        shading=ball
    },
    arrow/.style={->, thick}
]
         \node[comp, ball color=blue!45](Uu) at (0,0) {$U_{u}(t)$};
         \node[comp, ball color=orange!50](Us) at (5,0) {$U_{s}(t)$};
         \node[comp, ball color=red!55, below=3cm of Uu](D){$D(t)$};
         \node[comp, ball color=green!40 , below=3cm of Us](E){$E(t)$};
         \node[comp, ball color=yellow!50 , below right=1.5cm of Us, xshift=13mm](V){$V(t)$};
         \draw[->, thick, draw=blue!70!black](Uu.north) -- ++(0,1) node[midway,right]{$\alpha_{1}U_{u}(t)$};
        \draw[->, thick, draw=blue!70!black] 
    (-1.8,0) -- (Uu.west) 
    node[midway,above]{$\Gamma$};
          \draw[->, thick, draw=blue!70!black](Us.north) -- ++(0,1) node[midway,right]{$\alpha_{2} U_{s}(t)$};
         \draw[->, thick, draw=blue!70!black](D.south) -- ++(0,-1) node[midway,right]{$\alpha_{4} D(t)$};
         \draw[->, thick, draw=blue!70!black](E.south) -- ++(0,-1) node[midway,right]{$\alpha_{3} E(t)$};
         \draw[->, thick, draw=blue!70!black](Uu.east) -- (Us.west) node[midway,above]{$\beta_{1}U_{u}(t)U_{s}(t)$};
         \draw[->, thick, draw=blue!70!black](Uu.south) -- (D.north) node[midway,left]{$\beta_{2}U_{u}(t)D(t)$};
         \draw[->, thick, draw=blue!70!black](D.north east) -- (Us.south west) node[midway, sloped, above, fill=white, inner sep=1pt] {$\beta_{6}U_{s}(t)D(t)$};
         \draw[->, thick, draw=blue!70!black]([yshift=-3pt]D.east)--([yshift=-3pt]E.west) node[midway,below]{$\beta_{5}D(t)V(t)$};
         \draw[<-, thick, draw=blue!70!black]([yshift=3pt]D.east)--([yshift=3pt]E.west) node[midway,above]{$\beta_{4}E(t)$};
         \draw[->, thick, draw=blue!70!black](Us.south)--(E.north) node[midway,right]{$\beta_{3}U_{s}(t)V(t)$};
          \draw[->, thick, draw=blue!70!black](V.south) -- ++(0,-1) node[pos=1,above,yshift=-16pt] {$\delta V(t)$};
         \draw[
          <-, thick, draw=blue!70!black](V.north) -- ++(0,1) node[pos=1,above,yshift=1pt] {$\phi U_{s}(t)$};
     \end{tikzpicture}
    \caption{Schematic diagram of the model}
   \label{fig:figure1}
\end{figure}
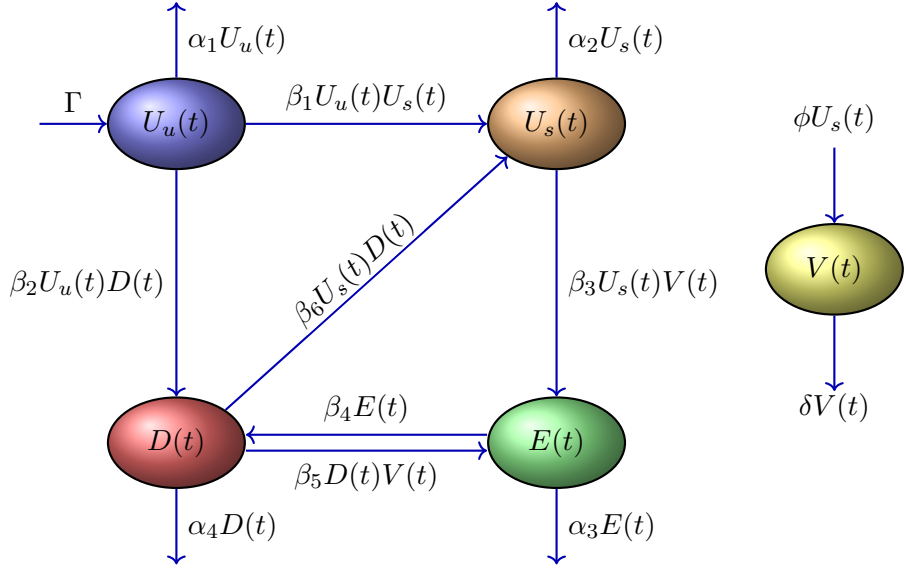
The model dynamics can be defined by the following system of differential equations.
    \begin{align}
        \dfrac{dU_{u}(t)}{dt}&=\Gamma-\beta_{1}U_{u}(t)U_{s}(t)-\beta_{2}U_{u}(t)D(t)-\alpha_{1}U_{u}(t)\nonumber\\
        \dfrac{dU_{s}(t)}{dt}&= \beta_{1}U_{u}(t)U_{s}(t)+\beta_{6}U_{s}(t)D(t)-\beta_{3}U_{s}(t) V(t)-\alpha_{2}U_{s}(t)\nonumber\\
        \dfrac{dE(t)}{dt}&=\beta_{3}U_{s}(t) V(t)+\beta_{5}D(t)V(t)-\beta_{4}E(t)-\alpha_{3}E(t)\label{Eq1}\\
        \dfrac{dD(t)}{dt}&=\beta_{2}U_{u}(t)D(t)-\beta_{6}U_{s}(t)D(t)-\beta_{5}D(t)V(t)+\beta_{4}E(t)-\alpha_{4}D(t)\nonumber\\
        \dfrac{dV(t)}{dt}&=\phi U_{s}(t) -\delta V(t),\nonumber
    \end{align}
with the initial conditions, $U_{u}(0)=U_{u_{0}}\ge0$, $U_{s}(0)=U_{s_{0}}\ge0$, $E(0)=E_{0}\ge0$, $D(0)=D_{0}\ge0$, $V(0)=V_{0}\ge0$. The variables and parameters are justified in the following table:

\begin{longtable}{p{3.5cm} p{11cm}}
\caption*{\MakeUppercase{Nomenclature}}\\
\hline
\textbf{Notation} & \textbf{Description} \\
\hline
$U_{u}(t)$ & Total number of unskilled and unemployed person at time $t$\\
$U_{s}(t)$ & Total number of skilled and unemployed person at time $t$\\
$D(t)$ & Total number of discouraged population at time $t$\\
$E(t)$ & Total number of employed person at time $t$ \\
$V(t)$ & Total number of vacancies at time $t$\\ 
$\alpha_{1}$ & Death rate of unskilled unemployed individuals\\
$\alpha_{2}$ & Death rate of skilled unemployed individuals\\
$\alpha_{3}$ & Death rate of employed individuals\\
$\alpha_{4}$ & Death rate of discouraged individuals\\ 
$\beta_{1}$& The rate at which unskilled unemployed individuals are moving towards the skilled unemployed class\\
$\beta_{2}$& The rate at which the total number of unskilled unemployed individuals is moving to the discouraged class \\
$\beta_{3}$&  The rate at which skilled individuals are being employed \\
$\beta_{4}$& The rate at which employed individuals are losing their jobs and moving to the discouraged class  \\
$\beta_{5}$& The rate at which discouraged individuals are moving to the employed class  \\
$\beta_{6}$& The rate at which discouraged individuals are moving to the skilled unemployed class\\
$\phi$& The rate at which vacancies are created for skilled unemployed individuals\\
$\delta$& The diminution rate of vacancies due to lack of funds\\ 
\hline
\end{longtable}

Subsequently, we prove the following theorem, which establishes the non-negativity and boundedness of the solutions of model \eqref{Eq1}.

\begin{theorem}
    The set 
    \begin{align*}
        \mathcal{S}=&\left\{\left(U_{u}(t),U_{s}(t),E(t), D(t),V(t) \right):0\leq U_{u}(t)+U_{s}(t)+E(t)+ D(t)\le \dfrac{\Gamma}{\gamma}; 0\le V\le \dfrac{\Gamma \phi}{\delta\gamma}  \right\},
    \end{align*}
    where $\gamma=\min\{\alpha_{1},\alpha_{2},\alpha_{3},\alpha_{4}\}$ is a region of attraction for the model in system \eqref{Eq1} and attracts all solutions initiated in the interior of the positive octant.
\end{theorem}
\begin{proof}
    From \eqref{Eq1} we have 
    \begin{equation}
        \notag
        \begin{aligned}
            \dfrac{d}{dt}\left\{U_{u}(t)+U_{s}(t)+E(t)+D(t)\right\}&=\Gamma-\alpha_{1}U_{u}(t)-\alpha_{2}U_{s}(t)-\alpha_{3}E(t)_-\alpha_4D(t)\\
            &\leq \Gamma-\gamma\left(U_{u}(t)+U_{s}(t)+E(t)+D(t) \right),
        \end{aligned}
    \end{equation}
     where $\gamma=\min\{\alpha_{1},\alpha_{2},\alpha_{3},\alpha_{4}\}$. Taking the limit supremum, we have
     \begin{equation}
         \notag
         \begin{aligned}
             \limsup\limits_{t\to\infty}\left\{ \Gamma-\gamma\left(U_{u}(t)+U_{s}(t)+E(t)+D(t) \right)\right\}\ge0.
         \end{aligned}
     \end{equation}
     This implies 
     \begin{equation}
         \notag
         \begin{aligned}
             \limsup\limits_{t\to\infty}\left\{ \left(U_{u}(t)+U_{s}(t)+E(t)+D(t) \right)\right\}\le \dfrac{\Gamma}{\gamma}.
         \end{aligned}
     \end{equation}
     Also from equation \eqref{Eq1} we have 
     \begin{equation}
         \notag
         \begin{aligned}
             \dfrac{dV}{dt}&=\phi U_{s}-\delta V\\
                           &\leq \dfrac{\Gamma \phi}{\gamma}-\delta V.
         \end{aligned}
     \end{equation}
     Taking the limit supremum, we have
     \begin{equation}
         \notag
          \limsup\limits_{t\to\infty}\left\{\dfrac{\Gamma \phi}{\gamma}-\delta V  \right\}\geq0.
     \end{equation}
     This implies,
     \begin{equation}
         \notag
         \begin{aligned}
             \limsup\limits_{t\to\infty} ~ V\leq \dfrac{\Gamma \phi}{\gamma\delta}.
         \end{aligned}
     \end{equation}
     This completes the proof.
\end{proof}

\section{Equilibrium Analysis}\label{S4}
In this section, we obtain equilibrium points by equating the right-hand side of the equation \eqref{Eq1} with zero. 
\begin{subequations}
\begin{align}
     \Gamma-\beta_{1}U_{u}U_{s}-\beta_{2}U_{u}D-\alpha_{1}U_{u}&=0, \label{Eq2}\\
     \beta_{1}U_{u}U_{s}+\beta_{6}U_{s}D-\beta_{3}U_{s} V-\alpha_{2}U_{s}&=0,\label{Eq3}\\
     \beta_{3}U_{s} V+\beta_{5}DV-\beta_{4}E-\alpha_{3}E&=0,\label{Eq4}\\
     \beta_{2}U_{u}D-\beta_{6}U_{s}D-\beta_{5}DV+\beta_{4}E-\alpha_{4}D&=0,\label{Eq5}\\
     \phi U_{s}-\delta V&=0.\label{Eq6}
\end{align}
\end{subequations}
From the equation \eqref{Eq2} we obtain
\begin{equation}
    U_{u}=\dfrac{\Gamma}{\alpha_{1}+\beta_{1} U_{s}+\beta_{2}D}\label{Eq7}.
\end{equation}
From equation \eqref{Eq6} we have 
\begin{equation}
    V=\dfrac{\phi U_{s} }{\delta}.\label{Eq8}
\end{equation}
From equations \eqref{Eq4} and \eqref{Eq8} we have
\begin{equation}
    \phi\left(\beta_{3}U_{s}^2+\beta_5 D U_{s} \right)-\delta(\beta_{4}+\alpha_{3})E=0.\label{Eq9}
\end{equation}
From equation \eqref{Eq3} we obtain 
$$U_{s} \left(\beta_{1}U_{u}+\beta_{6}D-\beta_{3} V-\alpha_{2}\right)=0.$$
This implies either
$$U_{s} =0,$$
or 
$$\beta_{1}U_{u}+\beta_{6}D-\beta_{3} V-\alpha_{2}=0.$$
{\it Case-I} If $U_{s}=0$ (i.e., no one is present in the skilled unemployed class). Then, from equations \eqref{Eq8} and \eqref{Eq9} we obtain
$$V=0 \quad\text{and} \quad E=0.$$
From equation equation \eqref{Eq5} we have
$$D(\beta_{2}U_{u}-\alpha_{4})=0.$$
This implies either
$$D=0,$$
or
$$U_{u}=\dfrac{\alpha_{4} }{\beta_{2}}.$$
If $$D=0,$$ then we get from \eqref{Eq7}
$$U_{u}=\dfrac{\Gamma}{\alpha_{1} }.$$
Therefore, one critical point is $$\left(U_{u},U_{s},E,D,V \right)=\left(\dfrac{\Gamma}{\alpha_{1}},0,0,0,0 \right). $$
\\If 
$$U_{u}=\dfrac{\alpha_{4} }{\beta_{2}},$$
then from equation \eqref{Eq7} we have
$$D=\dfrac{\Gamma\beta_{2}-\alpha_{1}\alpha_{4}}{\alpha_{4}\beta_{2}}. $$
Another critical point is 
$$\left(U_{u},U_{s},E,D,V \right)=\left(\dfrac{\alpha_{4}}{\beta_{2}},0,0,\dfrac{\Gamma\beta_{2}-\alpha_{1}\alpha_{4}}{\alpha_{4}\beta_{2}},0 \right), $$ and it will exist if $\Gamma\beta_{2}-\alpha_{1}\alpha_{4}>0.$
\\\textit{Case-II} If $U_{s}\neq0$, then from \eqref{Eq3}, we have
\begin{equation}
    \beta_{1}U_{u}+\beta_{6}D-\beta_{3} V-\alpha_{2}=0,\label{Eq10}
\end{equation}
Since $U_{s}\neq0$, then from equations \eqref{Eq6} and \eqref{Eq9} we can say that $V\neq0$ and $E\neq0$. Also $E\neq0$, then by using \eqref{Eq5} we have $D\neq0$ and hence, from \eqref{Eq7} $U_u\neq0$.
\\Therefore, from equations \eqref{Eq7}, \eqref{Eq8}, and \eqref{Eq10} we have
     \begin{align}
         0&=\dfrac{\beta_{1}\delta\Gamma}{\alpha_{1}+\beta_{1}U_{s}+\beta_{2}D}+\beta_{6}\delta D-\beta_{3}\phi U_{s}-\alpha_{2}\delta \notag\\
         &=\beta_{1}\delta\Gamma+ \beta_{6}\delta D\left(\alpha_{1}+\beta_{1}U_{s}+\beta_{2}D \right)-\beta_{3}\phi U_{s}\left(\alpha_{1}+\beta_{1}U_{s}+\beta_{2}D \right)-\alpha_{2}\delta\left(\alpha_{1}+\beta_{1}U_{s}+\beta_{2}D \right)\notag\\
        &= \delta\left(\beta_{1}\Gamma-\alpha_{1}\alpha_{2} \right)- \left(\alpha_{1}\beta_{3}\phi+\alpha_{2}\beta_{1}\delta \right)U_s+\delta\left(\alpha_{1}\beta_{6}-\alpha_{2}\beta_{2} \right)D +\left(\beta_{1}\beta_{6}\delta-\beta_{2}\beta_{3}\phi \right)U_{s}D \notag\\
        &~~~~+\beta_{2}\beta_{6}\delta D^2 -\phi\beta_{1}\beta_{3}U_{s} ^2.\label{Eq11}
     \end{align}
Also, substituting the values of \eqref{Eq7}, \eqref{Eq8}, and \eqref{Eq9} in \eqref{Eq5}, we have
     \begin{align}
         0=\dfrac{\beta_{2}\Gamma\delta D}{\alpha_{1}+\beta_{1}U_{s}+\beta_{2}D}-(\delta\beta_{6}+\phi\beta_{5})U_{s}D+\dfrac{\beta_{4}\phi U_{s}}{\alpha_{3}+\beta_{4}}(\beta_{3}U_{s}+\beta_{5}D)-\alpha_{4}\delta D.\label{Eq12}
     \end{align}
     Multiplying both side of the equation \eqref{Eq12} by $\left(\alpha_{1}+\beta_{1}U_{s}+\beta_{2}D\right)$ and $\left(\alpha_{3}+\beta_{4}\right)$ we obtain 
    \begin{align}
0 = & -\left\{ \alpha_{4}\beta_{2}\delta(\alpha_{3}+\beta_{4})+\beta_{2}\left(\beta_{4}\beta_{6}\delta+\alpha_{3}\beta_{6}\delta+\alpha_{3}\beta_{5}\phi \right)U_{s} \right\} D^{2} \notag \\
& + \big[ \delta\left(\beta_{2}\Gamma-\alpha_{1}\alpha_{4} \right)(\alpha_{3}+\beta_{4}) 
    - \left\{ \alpha_{1}\left(\beta_{4}\beta_{6}\delta+\alpha_{3}\beta_{6}\delta+\alpha_{3}\beta_{5}\phi \right)
    + \alpha_{4}\beta_{1}\delta(\alpha_{3}+\beta_{4}) \right\}U_{s}  \notag \\
&  - \left\{ \beta_{1}\left(\beta_{4}\beta_{6}\delta+\alpha_{3}\beta_{6}\delta+\alpha_{3}\beta_{5}\phi \right)-\beta_{2}\beta_{3}\beta_{4}\phi \right\}U_{s}^{2}~ \big]D + \alpha_{1}\beta_{3}\beta_{4}\phi U_{s}^{2}+\beta_{1}\beta_{3}\beta_{4}\phi U_{s}^{3}.\label{Eq13}
\end{align}
  The above equation \eqref{Eq13} can be written as
\begin{equation}
\label{Eq14}
    AD^{2}+BD+C=0,
\end{equation}
where $$
\begin{aligned}
A &= -\left\{ 
    \alpha_{4}\beta_{2}\delta(\alpha_{3}+\beta_{4})
    + \beta_{2}\left(\beta_{4}\beta_{6}\delta+\alpha_{3}\beta_{6}\delta+\alpha_{3}\beta_{5}\phi \right)U_{s} 
    \right\}, \\
 B &= \delta\left(\beta_{2}\Gamma-\alpha_{1}\alpha_{4} \right)(\alpha_{3}+\beta_{4})
    - \left\{\alpha_{1}\left(\beta_{4}\beta_{6}\delta+\alpha_{3}\beta_{6}\delta+\alpha_{3}\beta_{5}\phi \right)+\alpha_{4}\beta_{1}\delta(\alpha_{3}+\beta_{4}) 
      \right\}U_{s} \\
  &\quad 
    - \left\{\beta_{1}\left(\beta_{4}\beta_{6}\delta+\alpha_{3}\beta_{6}\delta+\alpha_{3}\beta_{5}\phi \right)
    -\beta_{2}\beta_{3}\beta_{4}\phi 
      \right\}U_{s}^{2},\\
    C&= \alpha_{1}\beta_{3}\beta_{4}\phi U_{s}^{2}+\beta_{1}\beta_{3}\beta_{4}\phi U_{s}^{3}.
\end{aligned}
$$
Then,
$$D=\dfrac{-B\pm \sqrt{B^{2}-4AC}}{2A}. \text{ (Provided $A\neq 0$)}$$
Substituting the value of $D$ in the equation \eqref{Eq11} we obtain 
\begin{equation}
\label{Eq15}
\begin{aligned}
    0=&\left\{\beta_{2}\beta_{6}\delta\left(B^{2}-2AC\right)+(\beta_{1}\beta_{6}\delta-\beta_{2}\beta_{3}\phi)ABU_{s}-2A^{2}\beta_{1}\beta_{3}\phi U_{s}^2+\delta\left(\alpha_{1}\beta_{6}-\alpha_{2}\beta_{2} \right)AB\right.\\
    &\left.-2A^2\left(\alpha_{1}\beta_{3}\phi+\alpha_{2}\beta_{1}\delta \right)U_{s}+2A^2\delta(\beta_{1}\Gamma-\alpha_{1}\alpha_{4})\right\}^2+\left\{\beta_{2}\beta_{6}\delta B+(\beta_{1}\beta_{6}\delta-\beta_{2}\beta_{3}\phi)AU_{s}\right.\\
    &\left.\left(\alpha_{1}\beta_{6}-\alpha_{2}\beta_{2} \right)A \right\}^2 \left(4AC-B^{2}\right).
\end{aligned}
\end{equation}
Let $$\begin{aligned}
        p&=\delta\left(\alpha_{3}+\beta_{4}\right),\\
        q&=\beta_{4}\beta_{6}\delta+\alpha_{3}\beta_{6}\delta+\alpha_{3}\beta_{5}\phi.
\end{aligned}$$
Then,
$$\begin{aligned}
    A&=-\left\{\alpha_{4}\beta_{2}p+\beta_{2}q U_{s} \right\},\\
    B&=p\left(\Gamma\beta_{2}-\alpha_{1}\alpha_{4} \right) -\left(\alpha_{1}q+\alpha_{4}\beta_{1}p \right)U_{s}-\left(\beta_{1}q-\beta_{2}\beta_{3}\beta_{4}\phi \right)U_{s}^{2},\\
    C&=\beta_{3}\beta_{4}\phi\left(\alpha_{1}+\beta_{1}U_{s} \right) U_{s}^{2}.
\end{aligned}$$
Therefore, equation \eqref{Eq15} can be written as 
\begin{equation}\label{Eq16}
    a_{0}U_{s}^{8}+a_{1}U_{s}^{7}+a_{2}U_{s}^{6}+a_{3}U_{s}^{5}+a_{4}U_{s}^{4}+a_{5}U_{s}^{3}+a_{6}U_{s}^{2}+a_{7}U_{s}+a_{8}=0,
\end{equation}
   where the coefficients $a_{0}$ and $a_{8}$  are given by
       \begin{align*}
a_{0}=&-8q^2\beta_{1}\beta_{2}^{3}\beta_{3}\phi\left(q+\beta_{4}\beta_{6}\delta\right)\left(\beta_{1}\beta_{6}\delta-\beta_{2}\beta_{3}\phi\right)\left(q\beta_{1}-\beta_{2}\beta_{3}\beta_{4}\phi \right),\\
a_{8}=&4p^{4}\alpha_{4}^{2}\beta_{2}^{3}\delta^{2}\left(\Gamma\beta_{1}-\alpha_{1}\alpha_{2}\right)\left(\Gamma^{2}\beta_{2}^{2}\beta_{6}+\Gamma\alpha_{2}\alpha_{4}\beta_{2}^2+\Gamma\alpha_{4}^{2}\beta_{1}\beta_{2}+2\alpha_{1}^{2}\alpha_{4}^{2}\beta_{6}-2\alpha_{1}\alpha_{2}\alpha_{4}^{2}\beta_{2}\right.\\
&\left.-3\Gamma\alpha_{1}\alpha_{4}\beta_{2}\beta_{6} \right).
\end{align*}
If the product of the leading coefficient and the constant term of \eqref{Eq16} is negative, then a positive root of \eqref{Eq15} must exist. If \eqref{Eq16} has a positive root (say $U_{s}^{*}$), then $A<0$ and $C>0$, hence $B^{2}-4AC>0$. Therefore, there will exist a positive $D$ (say $D^{*}$), which will be of the form 
   $$D^{*}=\dfrac{-B+\sqrt{B^{2}-4AC}}{2A}.$$
   Using this values of $U_{s}^{*}$ and $D^{*}$ we will obtain positive values of $U_{u}$, $V$, and $E$ from equations \eqref{Eq7}, \eqref{Eq8}, and \eqref{Eq9} respectively. Let us denote them by $U_{u}^{*}$, $V^{*}$, and $E^{*}$. Therefore, the critical points are 
   \begin{equation}
       \notag
       \begin{aligned}
           P_{0}&=\left(\dfrac{\Gamma}{\alpha_{1}},0,0,0,0 \right),\\
           P_{1}&=\left(\dfrac{\alpha_{4}}{\beta_{2}},0,0,\dfrac{\Gamma\beta_{2}-\alpha_{1}\alpha_{4}}{\alpha_{4}\beta_{2}},0 \right), \text{ exist if $\beta_{2}\Gamma-\alpha_{1}\alpha_{4}>0$},\\
           P_{2}&=\left(U_{u}^{*},U_{s}^{*},E^{*},D^{*}, V^{*} \right), \text{ exist if the product of the leading coefficient}\\
           &~~~~~~~~~~~~~~~~~~~~~~~~~~~~~~~~~
           \text{and the constant term of \eqref{Eq16} is negative.}
       \end{aligned}
   \end{equation}
   \begin{remark}
From equation \eqref{Eq16}, we obtain an $8^{\text{th}}$-degree polynomial equation
\begin{align*}
F(U_s) = 0,
\end{align*}
where
\begin{align*}
F(U_s) = a_0 U_s^8 + a_1 U_s^7 + a_2 U_s^6 + a_3 U_s^5 + a_4 U_s^4 + a_5 U_s^3 + a_6 U_s^2 + a_7 U_s + a_8 .
\end{align*}
Evaluating at $U_s = 0$ gives
\begin{align*}
F(0) = a_8.
\end{align*}
Moreover, $F(U_s) \to +\infty$ as $U_s \to \infty$ if $a_0 > 0$, and $F(U_s) \to -\infty$ if $a_0 < 0$. Hence, if $a_8$ and $a_0$ have opposite signs, then by the Intermediate Value Theorem, $F(U_s)$ admits at least one positive real root.
\end{remark}

\section{Stability Analysis}\label{S5}
In this section, we investigate the local stability of the equilibrium points. In order to check the local stability of the critical point, we compute the following variational matrix:
\scriptsize{\begin{equation}
    \label{Eq17}
    M=\begin{pmatrix}
        -\left(\alpha_{1}+\beta_{1} U_{s}+\beta_{2} D \right) & -\beta_{1}U_{u} & 0& -\beta_{2}U_{u} &0\\[2ex]
        \beta_{1} U_{s} & -(\alpha_{2}+\beta_{3}V)+\beta_{1}U_{u}+\beta_{6}D &0 &\beta_{6}U_{s}  &-\beta_{3}U_{s}\\[2ex]
        0 & \beta_{3} V & -(\alpha_{3}+\beta_{4}) & \beta_{5}V & \beta_{3}U_{s}+\beta_{5}D\\[2ex]
        \beta_{2} D &-\beta_{6} D &\beta_{4} &-(\alpha_{4}+\beta_{5}V+\beta_{6}U_{s})+\beta_{2}U_{u} &-\beta_{5}D &\\[2ex]
        0 & \phi & 0 & 0 & -\delta
    \end{pmatrix}.
\end{equation}}
\normalsize
\begin{enumerate}
    \item[$(i)$] At the point $P_{0}=\left(\dfrac{\Gamma}{\alpha_{1}},0,0,0,0 \right)$.
\par The corresponding variational matrix is given by
\begin{equation}
    \notag
    M_{P_{0}}=\begin{pmatrix}
        -\alpha_{1} & -\dfrac{\Gamma \beta_{1}}{\alpha_{1}} & 0 & -\dfrac{\Gamma \beta_{2}}{\alpha_{1}} & 0\\[2ex]
        0&\dfrac{\Gamma \beta_{1}-\alpha_{1}\alpha_{2}}{\alpha_{1}} &0 &0 &0 \\[2ex]
        0 &0 &  -(\alpha_{3}+\beta_{4})&0 &0\\[2ex]
        0 & 0 & \beta_{4} & \dfrac{\Gamma\beta_{2}-\alpha_{1}\alpha_{4}}{\alpha_{1}} &0\\[2ex]
        0 &\phi &0 &0 &-\delta
    \end{pmatrix}.
\end{equation}
The eigenvalues of the matrix  $M_{P_{0}}$ are given by
$$\lambda^{0}_{0}=-\alpha_{1},~ \lambda^{0}_{1}=-(\alpha_{3}+\beta_{4}),~\lambda^{0}_{2}=\dfrac{\Gamma\beta_{1}-\alpha_{1}\alpha_{2}}{\alpha_{1}},~\lambda^{0}_{3}=\dfrac{\Gamma\beta_{2}-\alpha_{1}\alpha_{4}}{\alpha_{1}},~\lambda^{0}_{4}=-\delta.$$
If both $\Gamma\beta_{1}-\alpha_{1}\alpha_{2}$ and $\Gamma\beta_{2}-\alpha_{1}\alpha_{4}$ are negative, then the critical point $P_{0}$ is stable.
\item[$(ii)$] At the point  $ P_{1}=\left(\dfrac{\alpha_{4}}{\beta_{2}},0,0,\dfrac{\Gamma\beta_{2}-\alpha_{1}\alpha_{4}}{\alpha_{4}\beta_{2}},0 \right)$, exist if $\beta_{2}\Gamma-\alpha_{1}\alpha_{4}>0$.
\par The corresponding variational matrix is given by 
\footnotesize{\begin{equation}
    \notag
    M_{P_{1}}=\begin{pmatrix}
         -\alpha_{1}-\dfrac{\Gamma\beta_{2}-\alpha_{1}\alpha_{4}}{\alpha_{4}} & -\dfrac{\alpha_{4}\beta_{1}}{\beta_{2}} &0 &-\alpha_{4} &0\\[2ex]
         0& -\alpha_{2}+\dfrac{\alpha_{4}\beta_{1}}{\beta_{2}}+\dfrac{\beta_{6}(\Gamma\beta_{2}-\alpha_{1}\alpha_{4})}{\alpha_{4}\beta_{2}} & 0 & 0 &0\\[2ex]
         0 & 0 &-(\alpha_{3}+\beta_{4}) &0 & \dfrac{\beta_{5}(\Gamma\beta_{2}-\alpha_{1}\alpha_{4})}{\alpha_{4}\beta_{2}}\\[2ex]
         \dfrac{\Gamma\beta_{2}-\alpha_{1}\alpha_{4}}{\alpha_{4}} & -\dfrac{\beta_{6}(\Gamma\beta_{2}-\alpha_{1}\alpha_{4})}{\alpha_{4}\beta_{2}} &\beta_{4} &0 &-\dfrac{\beta_{5}(\Gamma\beta_{2}-\alpha_{1}\alpha_{4})}{\alpha_{4}\beta_{2}}\\[2ex]
         0 &\phi &0 &0 &-\delta
    \end{pmatrix}.
\end{equation}}
\normalsize
 The eigenvalues are given by
\begin{equation}
    \notag
    \begin{aligned}
        \lambda^{1}_{0}&=-(\alpha_{3}+\beta_{4}), ~ \lambda_{1}^{1}=-\delta,~ \lambda_{2}^{1}=-\alpha_{2}+\beta_{6}x+y, ~\lambda_{3}^{1}=-\frac{1}{2}\left(\alpha_{1}+\beta_{2}x+\sqrt{\left(\alpha_{1}+\beta_{2}x\right)^2-4\alpha_{4}\beta_{2}x} \right),\\
        \lambda_{4}^{1}&=-\frac{1}{2}\left(\alpha_{1}+\beta_{2}x-\sqrt{\left(\alpha_{1}+\beta_{2}x\right)^2-4\alpha_{4}\beta_{2}x} \right),
    \end{aligned}
\end{equation}
where $x=\dfrac{\Gamma\beta_{2}-\alpha_{1}\alpha_{4} }{\alpha_{4}\beta_{2}}$ and $y=\dfrac{\alpha_{4}\beta_{1}}{\beta_{2}}$. Therefore, the critical point $P_{1}$ is stable if $\lambda_{2}^{1}$ is negative.
\item[$(iii)$]  At the point $ P_{2}=\left(U_{u}^{*}, U_{s}^{*}, D^{*}, E^{*}, V^{*} \right)$, exists if the product of the leading coefficient and the constant term in \eqref{Eq16} is negative.
\\The corresponding variational matrix is given by
\tiny{\begin{equation}
\notag
    M_{P_{2}}=\begin{pmatrix}
        -\left(\alpha_{1}+\beta_{1} U_{s}^{*}+\beta_{2} D^{*} \right) & -\beta_{1}U_{u}^{*} & 0& -\beta_{2}U_{u}^{*} &0\\[2ex]
        \beta_{1} U_{s}^{*} & -(\alpha_{2}+\beta_{3}V^{*})+\beta_{1}U_{u}^{*}+\beta_{6}D^{*} &0 &\beta_{6}U_{s}^{*}  &-\beta_{3}U_{s}^{*}\\[2ex]
        0 & \beta_{3} V^{*} & -(\alpha_{3}+\beta_{4}) & \beta_{5}V^{*} & \beta_{3}U_{s}^{*}+\beta_{5}D^{*}\\[2ex]
        \beta_{2} D^{*} &-\beta_{6} D^{*} &\beta_{4} &-(\alpha_{4}+\beta_{5}V^{*}+\beta_{6}U_{s}^{*})+\beta_{2}U_{u}^{*} &-\beta_{5}D^{*} \\[2ex]
        0 & \phi & 0 & 0 & -\delta
    \end{pmatrix}.
\end{equation}}
\normalsize
The characteristic equation of the matrix $M_{P_{2}}$ is given by
\begin{equation}
\label{Eq18}
    \eta^{5}+c_{1}\eta^{4}+c_{2}\eta^{3}+c_{3}\eta^{2}+c_{4}\eta+c_{5}=0,
\end{equation}
where
    \begin{align}
        c_{1}=&x+y+w+\alpha_{3}+\beta_{4}+\delta-\beta_{1}U_{u}^{*}-\beta_{2}U_{u}^{*}-\beta_{6}D^{*},\nonumber\\
        c_{2}=& xy - \beta_{1} U_{u}^{*}(x+y)  + \beta_{1}^{2} U_{s}^{*} U_{u}^{*} -x \beta_{2} U_{u}^{*}  +\beta_{1} \beta_{2}\left(U_{u}^{2}\right)^{*}+\beta_{2}^{2}U_{u}^{*}D^{*} 
        - \beta_{4} \beta_{5} V^{*} - x \beta_{6} D^{*} -y \beta_{6} D^{*} \nonumber  \\
& + \beta_{2} \beta_{6} D^{*} U_{u}^{*} + \beta_{6}^{2}D^{*} U_{s}^{*} 
+ x \delta + y \delta -\beta_{1} \delta U_{u}^{*} 
-\beta_{2} \delta U_{u}^{*}  -\beta_{6} \delta D^{*} \nonumber \\
& + w (x + y - \beta_{2}U_{u}^{*}  + \delta)
+ p \left(w + x + y -\beta_{1} U_{u}^{*}  -\beta_{2} U_{u}^{*} 
-\beta_{6} D^{*}  + \delta\right)
+ \beta_{3} \phi U_{s}^{*} ,\nonumber\\
c_{3}=&pxw+xyw-px\beta_{1} U_{u}^{*}+p\beta_{1}^{2}U_{s}^{*}U_{u}^{*}+y\beta_{1}^{2}U_{s}^{*}U_{u}^{*}-xw\beta_{2}U_{u}^{*}+x\beta_{1}\beta_{2}\left(U_{u}^{*} \right)^{2}-\beta_{1}^{2}\beta_{2}\left(U_{u}^{*} \right)^{2}U_{s}^{*}  \nonumber \\
&+\beta_{1}\beta_{2}\beta_{3} U_{u}^{*}U_{s}^{*}V^{*}-\beta_{4}\beta_{5} V^{*}\left(x+w-\beta_{1}U_{u}^{*} \right)-px\beta_{6}D^{*}-xy\beta_{6}D^{*}+x\beta_{2}\beta_{6}U_{u}^{*}D^{*},
   \nonumber\\
   &+\beta_{1}\beta_{2}\beta_{6}U_{u}^{*}U_{s}^{*}D^{*}-\beta_{3}\beta_{4}\beta_{6}U_{s}^{*}V^{*}+p\beta_{6}^{2}U_{s}^{*}D^{*}+x\beta_{6}^{2}U_{s}^{*}D^{*}+py\left(w+x-\beta_{1}U_{u}^{*}-\beta_{6}D^{*} \right) \nonumber\\
   &-p\beta{2}U_{u}^{*}\left(w+x-\beta_{1}U_{u}^{*}-\beta_{6}D^{*} \right)+\beta_{5}V^{*}D^{*}\left(\beta_{2}^{2}U_{u}^{*}+\beta_{4}\beta_{6} \right)+xw\delta+py\delta-x\delta\beta_{1}U_{u}^{*} \nonumber\\
   &+\beta_{1}^{2}\delta U_{u}^{*}U_{s}^{*}-p\beta_{2}\delta U_{u}^{*}-\beta_{4}\beta_{5}\delta V^{*}-x\beta_{6}\delta D^{*}+\beta_{6}^{2}\delta U_{u}^{*} D^{*}+p\delta\left(w+x-\beta_{1}U_{u}^{*}-\beta_{6}D^{*} \right)  \nonumber\\
   &+y\delta\left(w+x-\beta_{1}U_{u}^{*}-\beta_{6}D^{*} \right)-\beta_{2}\delta U_{u}^{*}\left(w+x-\beta_{1}U_{u}^{*}-\beta_{6}D^{*} \right)+p\beta_{4}\phi U_{s}^{*}+x\beta_{3}\phi U_{s}^{*} \nonumber\\
   &+\phi U_{s}^{*}\left(y\beta_{3}-\beta_{2}\beta_{3}U_{u}^{*}+\beta_{5}\beta_{6} D^{*} \right),  \nonumber\\
   c_{4}=&\beta_{2}^{2} \beta_{5} U_{u}^{*} D^{*} V^{*} w
- \beta_{1} \beta_{2}^{2} \beta_{5} U_{u}^{*2} D^{*} V^{*}
- \beta_{4} \beta_{5} V^{*} w x
+ \beta_{1} \beta_{4} \beta_{5} U_{u}^{*} V^{*} x \nonumber \\ 
& - \beta_{2}^{2} \beta_{5} \beta_{6} U_{u}^{*} D^{*2} V^{*}
+ \beta_{4} \beta_{5} \beta_{6} D^{*} V^{*} x
+ w x y \delta
- \beta_{1} U_{u}^{*} x y \delta
- \beta_{2} U_{u}^{*} w x \delta \nonumber \\
& + \beta_{1} \beta_{2} U_{u}^{*2} x \delta
+ \beta_{2}^{2} \beta_{5} U_{u}^{*} D^{*} V^{*} \delta
- \beta_{4} \beta_{5} V^{*} w \delta
- \beta_{4} \beta_{5} V^{*} x \delta
+ \beta_{1} \beta_{4} \beta_{5} U_{u}^{*} V^{*} \delta \nonumber \\
& - \beta_{6} D^{*} x y \delta
+ \beta_{2} \beta_{6} D^{*} U_{u}^{*} x \delta
+ \beta_{4} \beta_{5} \beta_{6} D^{*} V^{*} \delta
- U_{s}^{*} 
\Big(
- \beta_{6}^{2} D^{*} x \delta
+ \beta_{1} \beta_{2} U_{u}^{*2} (V^{*} \beta_{2} \beta_{3} + \beta_{1} \delta) \nonumber \\
& \quad - x y \beta_{3} \phi
+ z \beta_{4} \beta_{6} \phi
- D^{*} x \beta_{5} \beta_{6} \phi
- U_{u}^{*} (y \beta_{1}^{2} \delta
+ D^{*} \beta_{1} \beta_{2} \beta_{6} \delta
+ V^{*} (y \beta_{1} \beta_{2} \beta_{3} - \beta_{1}^{2} \beta_{4} \beta_{5}  \nonumber \\
& \quad + D^{*} \beta_{2}^{2} \beta_{3} \beta_{6}
- D^{*} \beta_{1} \beta_{2} \beta_{5} \beta_{6}
+ \beta_{1} \beta_{2} \beta_{3} \delta)
+ z \beta_{1} \beta_{2} \phi
- x \beta_{2} \beta_{3} \phi)
+ V^{*} \beta_{3} \beta_{4} (x \beta_{6} + \beta_{6} \delta + \beta_{5} \phi)
\Big) \nonumber \\
& + p 
\Big(
- D^{*} x y \beta_{6}
+ D^{*} U_{s}^{*} x \beta_{6}^{2}
+ x y \delta
- D^{*} x \beta_{6} \delta
- D^{*} y \beta_{6} \delta
+ D^{*} U_{s}^{*} \beta_{6}^{2} \delta
+ \beta_{1} \beta_{2} U_{u}^{*2} (x - U_{s}^{*} \beta_{1} + \delta) \nonumber \\
& \quad + w ((y - U_{u}^{*} \beta_{2}) \delta + x (y - U_{u}^{*} \beta_{2} + \delta))
+ U_{s}^{*} x \beta_{3} \phi
+ U_{s}^{*} y \beta_{3} \phi
+ D^{*} U_{s}^{*} \beta_{5} \beta_{6} \phi \nonumber \\
& \quad + U_{u}^{*} (-y \beta_{1} \delta
+ D^{*} \beta_{2} \beta_{6} \delta
- x (y \beta_{1} - D^{*} \beta_{2} \beta_{6} + (\beta_{1} + \beta_{2}) \delta)
\nonumber\\&+ U_{s}^{*} (y \beta_{1}^{2} + D^{*} \beta_{1} \beta_{2} \beta_{6} + \beta_{1}^{2} \delta - \beta_{2} \beta_{3} \phi))
\Big) , \nonumber \\
   c_{5}=& \Big[ 
(-y +\beta_{2} U_{u}^{*} ) 
\Big( 
- \beta_{1} \beta_{2} \beta_{3}  U_{s}^{*} U_{u}^{*} V^{*}
- p ( x w-x \beta_{1} U_{u}^{*}  + \beta_{1}^{2} U_{s}^{*} U_{u}^{*}  -x \beta_{6} D^{*} )
\Big) \nonumber \\
& +\beta_{6} U_{s}^{*}  
\Big(
-x \beta_{3} \beta_{4} V^{*}  
+ D^{*} \left(p \beta_{1} \beta_{2}  U_{u}^{*} 
+  \beta_{2}^{2} \beta_{3} U_{u}^{*} V^{*} 
+ p x \beta_{6}\right)
\Big) \nonumber \\
& -\beta_{5} V^{*}  
\Big(
(xw + \beta_{1} U_{u}^{*} (-x +\beta_{1} U_{s}^{*} )) \beta_{4} 
+ \beta_{2}^{2} \beta_{6}\left(D^{*}\right)^{2} U_{u}^{*}  \nonumber \\
& \quad - D^{*}(-\beta_{1} \beta_{2}^{2} (U_{u}^{*})^{2} 
+ x \beta_{4} \beta_{6} 
+\beta_{2} U_{u}^{*} (w \beta_{2} -\beta_{1} \beta_{6} U_{s}^{*} ))
\Big)
\Big] \delta \nonumber \\
& + \Big[\beta_{2}
U_{s}^{*} U_{u}^{*}  
\Big(
y z \beta_{1} 
-z \beta_{1} \beta_{2} U_{u}^{*}  
+\beta_{2} \beta_{3} \beta_{5}  D^{*} V^{*} 
- \beta_{1} \beta_{5}^{2} D^{*} V^{*} 
+ z \beta_{2} \beta_{6} D^{*}
\Big) \nonumber \\
& \quad +x U_{s}^{*}  
\Big(
- \beta_{4}(V^{*} \beta_{3} \beta_{5} + z \beta_{6}) 
+ p (y \beta_{3} -\beta_{2} \beta_{3} U_{u}^{*}  +\beta_{5}\beta_{6} D^{*}  )
\Big)
\Big] \phi , \nonumber
    \end{align}
and $x,y,z,w$ and $p$ are given by
\begin{align*}
    x&=\alpha_{1}+\beta_{1}U_{s}^{*}+\beta_{2}D^*,\\
        y&=\alpha_{4}+\beta_{5}V^{*}+\beta_{6}U_{s}^{*},\\
        z&=\beta_{3}U_{s}^{*}+\beta_{5}D^{*},\\
        w&=\alpha_{2}+\beta_{3}V^{*},\\
        p&=\alpha_{3}+\beta_{4}.
\end{align*}
To check the stability of the critical point $P_{2}$, we use the Routh-Hurwitz criterion. The RH (Routh-Hurwitz) matrix is given by
\begin{align*}
     RH_{5}&=\begin{pmatrix}
            c_{1}&1&0&0&0 \\
            c_{3}&c_{2}&c_{1}&1&0\\
            c_{5}&c_{4}&c_{3}&c_{2}&c_{1} \\
            0& 0&c_{5}&c_{4}&c_{3} \\
            0&0&0&0&c_{5}
        \end{pmatrix}.
\end{align*}

If each $c_{i},~ (i=1,2,3,4,5)$ are positive, also
\begin{align}
    &c_{1}c_{2}-c_{3}>0,\notag\\
    &c_{1}c_{2}c_{3}- c_{3}^{2}-c_{1}^{2}c_{4}+c_{1} c_{5}>0,\notag \\
    &c_1 c_2 c_3 c_4 - c_3^2 c_4 - c_1^2 c_4^2 - c_1 c_2^2 c_5 + c_2 c_3 c_5 + 2 c_1 c_4 c_5 - c_5^2
>0,\notag\\
&c_5 \left( c_1 c_2 c_3 c_4 - c_3^2 c_4 - c_1^2 c_4^2 - c_1 c_2^2 c_5 + c_2 c_3 c_5 + 2 c_1 c_4 c_5 - c_5^2 \right)>0,\notag
\end{align}
then the critical point $P_{2}$ is locally asymptotically stable.

\end{enumerate}

\section{Numerical Simulation}\label{S6}
This section consists of two parts: numerical solution plots illustrating the theoretical results and qualitative behaviour of the model, and a sensitivity analysis examining the impact of parameter variations on the state variables.

\subsection{Solution Plots} \

\parindent=0mm \vspace{.01in}
For numerical simulation, we have used MATLAB 2022a by considering the following parameters
\begin{align*}
    &\Gamma = 5000, ~ \alpha_{1} = 0.7,~ \alpha_{2} = 0.1,~ \alpha_{3} =  0.5,~ \alpha_{4} = 1.2, ~ \phi = 1.1, ~ \delta = 0.5,\\
    &\beta_{1} = 1.105,~ \beta_{2} = 1.0,~ \beta_{3} = 0.9,~ \beta_{4} = 0.05,~\beta_{5} = 0.7,~ \beta_{6} = 0.9.
\end{align*}
Therefore, the equilibrium values of the model \eqref{Eq1} corresponding to the above parameters are given by
\begin{align*}
    &P_{0}: U_{u}^{*}=7142, ~ U_{s}^{*}=0, ~ E^{*}=0, ~ D^{*}=0, ~ V^{*}=0,\\
    &P_{1}: U_{u}^{*}=1, ~ U_{s}^{*}=0, ~ E^{*}=0, ~ D^{*}=4166, ~ V^{*}=0,\\
    &P_{2}: U_{u}^{*}=76, ~ U_{s}^{*}=48, ~ E^{*}=9855, ~ D^{*}=12, ~ V^{*}=105.
\end{align*}
The corresponding eigenvalues are given by 
\begin{align*}
    & P_{0} : 7891.81,~ 7140.8,~ -0.7,~ -0.55,~ -0.5,\\
    & P_{1} : -4165.7,~ 3750.41,~ -1.20012,~ -0.55,~ -0.5,\\
    & P_{2} : -35.3691 + 68.8336 i,~ -35.3691 - 68.8336 i,~ -36.0122,~ -1.2591,~ -0.500451.
\end{align*}
Since all eigenvalues corresponding to the equilibrium point $P_{2}$ have negative real parts, the critical point $P_{2}$ is locally asymptotically stable.
\begin{figure}[H]
\begin{center}
\includegraphics[width=0.9\textwidth]{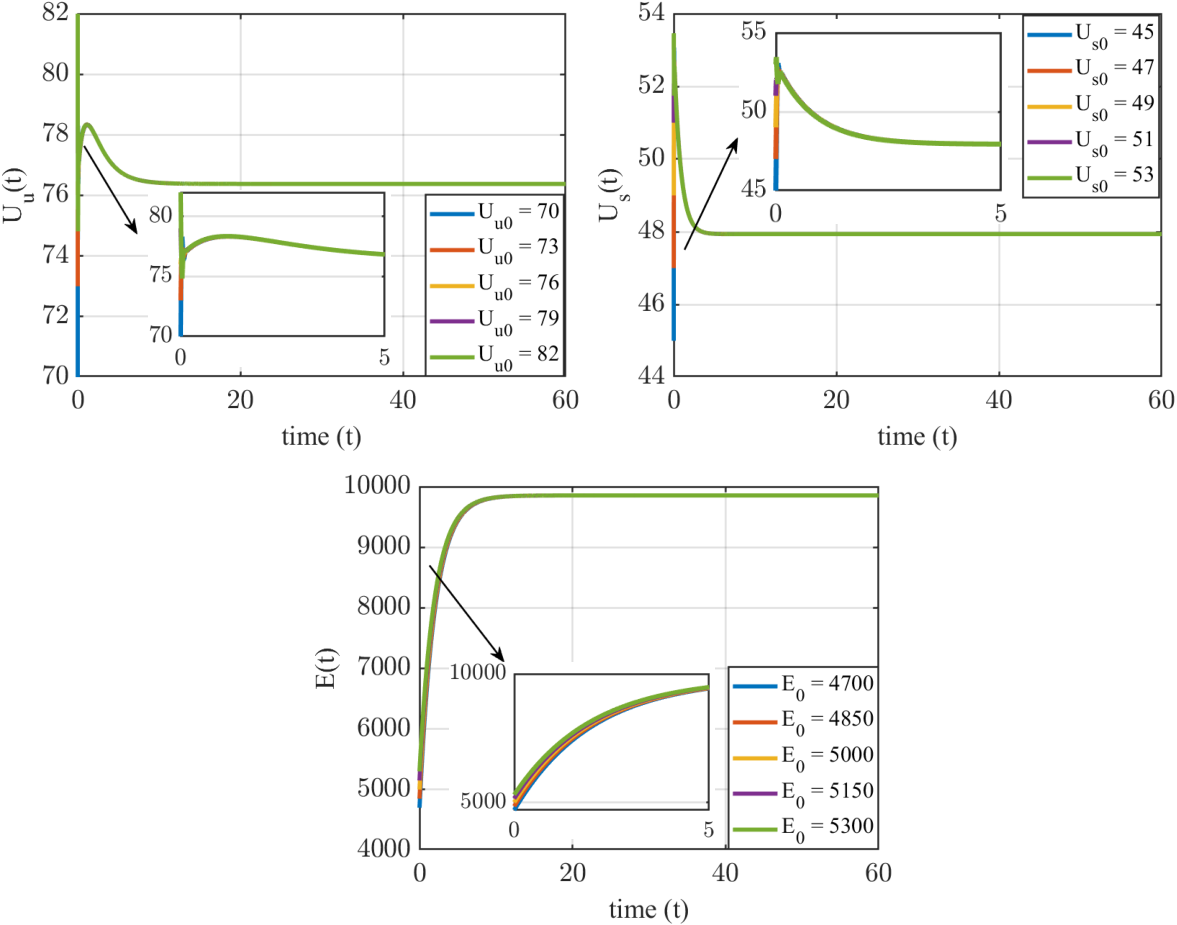}
\includegraphics[width=0.9\textwidth]{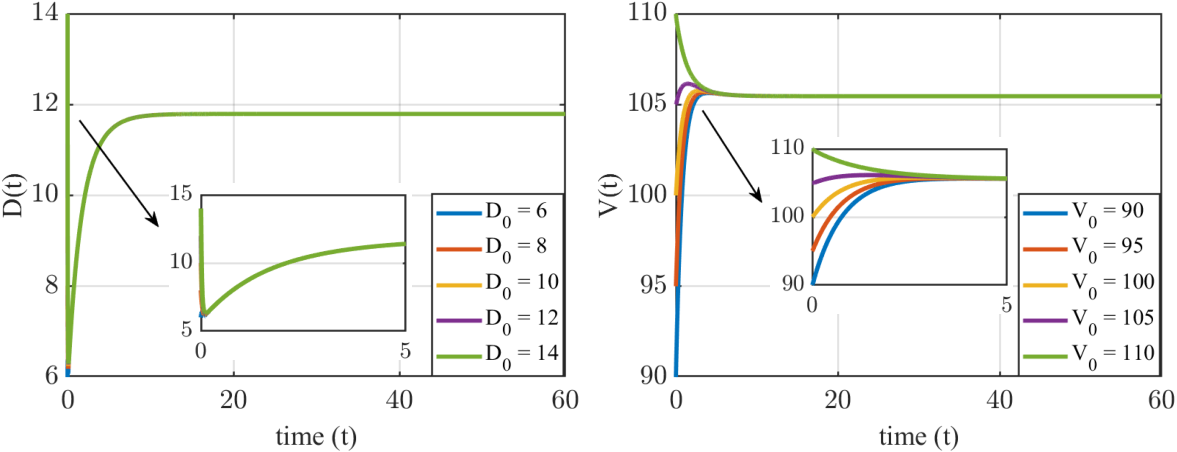}
\end{center}
\caption{Numerical simulation of the solution plots of $U_{u}(t), U_{s}(t), E(t), D(t),$ and $V(t)$ with respect to time (in days) with various initial conditions $(U_{u_{0}}, U_{s_{0}}, E_{0}, D_{0}, V_{0})$.}
\label{fig:figure2}
\end{figure}
From Fig.~\ref{fig:figure2}, it is observed that for different initial conditions, all state variables converge to their equilibrium values, indicating that the equilibrium point is locally asymptotically stable.
\subsection{Sensitivity Analysis}\

\parindent=8mm \vspace{.01in}
In this subsection, we analyze the sensitivity of the proposed model \eqref{Eq1} to assess how variations in key parameters influence the system dynamics and the behavior of the state variables. Specifically, we examine the response of the solution components $\big(U_{u}(t), U_{s}(t), E(t), D(t), V(t)\big)$ under $\pm 20\%$ variations in each principal parameter $\beta_i$, thereby evaluating the impact of parameter changes and the robustness of the model.

\begin{figure}[H]
\begin{center}
\includegraphics[width=0.9\textwidth]{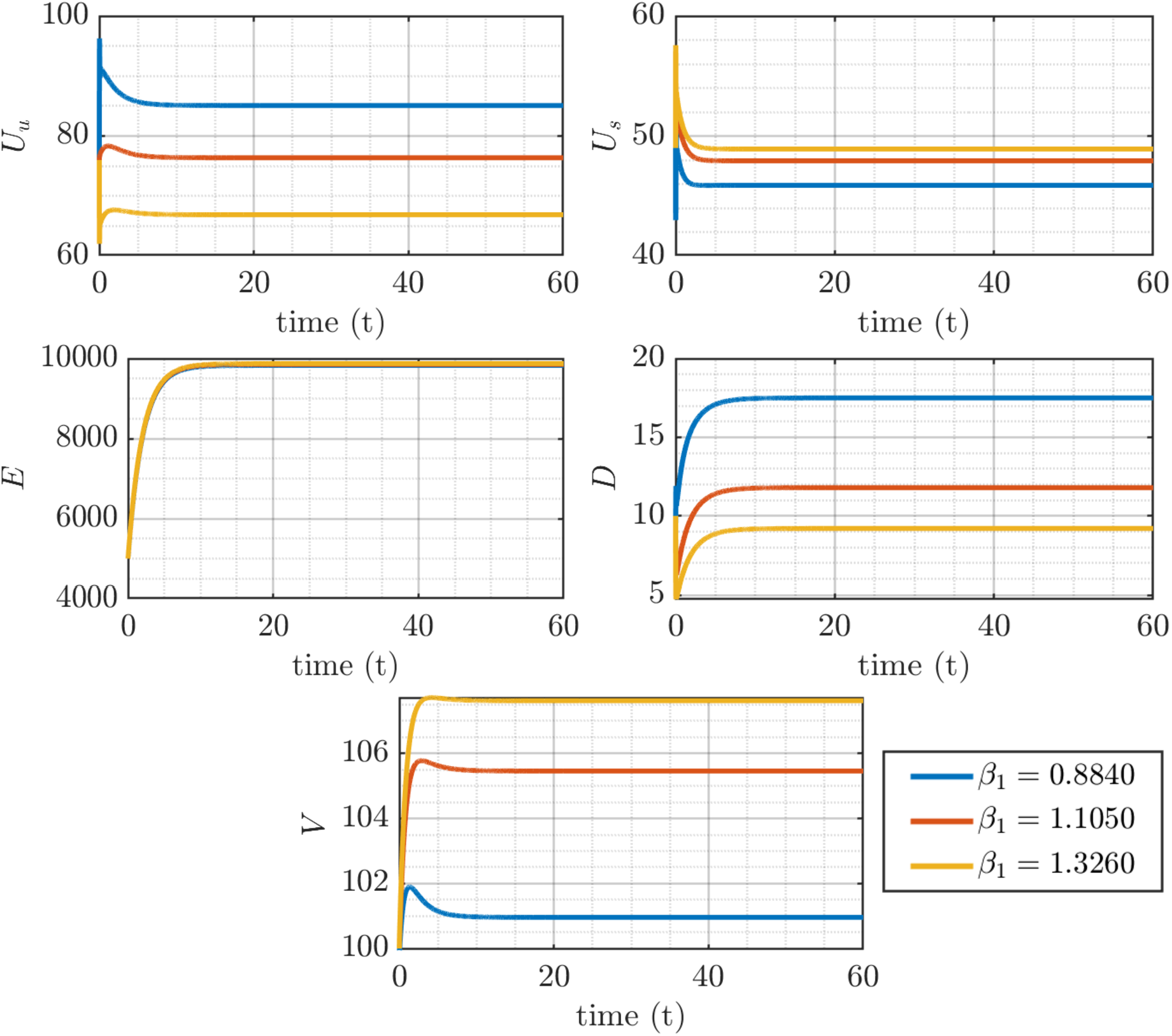}
\end{center}
\caption{Temporal dynamics of the $U_{u}(t), U_{s}(t), E(t), D(t),$ and $V(t)$ with respect to time (in days) with various values of $\beta_{1}.$}
\label{fig:figure3}
\end{figure}
From Fig~\ref{fig:figure3}, we observe:
\begin{enumerate}[(a)]
    \item For different values of the parameter $\beta_{1}$, all state variables asymptotically converge to a stable equilibrium point as time progresses.
    \item An increase in the parameter $\beta_{1}$ leads to a decrease in the number of unskilled unemployed individuals, whereas a reduction of $\beta_{1}$ results in an increase in this population.
    \item Conversely, as $\beta_{1}$ increases, the number of skilled unemployed individuals rises, and it declines when $\beta_{1}$ is reduced.
    \item The employed population $E(t)$ converges to nearly the same steady-state level for all considered values of $\beta_{1}$, indicating that employment is only weakly sensitive to variations in this parameter.
    \item The number of discouraged individuals decreases with an increase in $\beta_{1}$ and increases when $\beta_{1}$ decreases.
    \item An increase in $\beta_{1}$ results in a higher number of job vacancies, while a decrease in $\beta_{1}$ leads to a reduction in vacancies.
\end{enumerate}

\begin{figure}[H]
\begin{center}
\includegraphics[width=\textwidth]{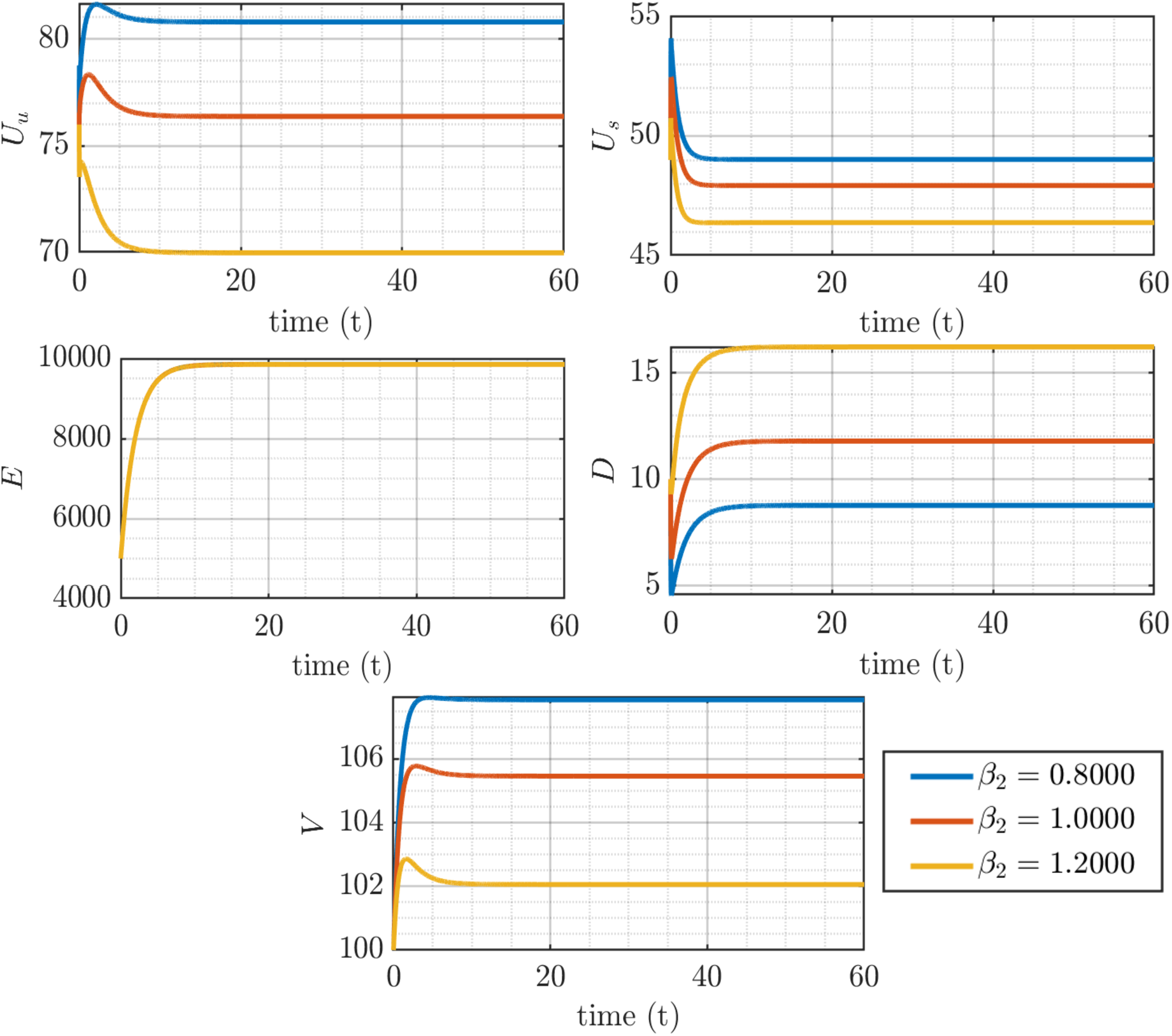}
\end{center}
\caption{Temporal dynamics of the $U_{u}(t), U_{s}(t), E(t), D(t),$ and $V(t)$ with respect to time (in days) with different values of $\beta_{2}.$}
\label{fig:figure4}
\end{figure}
From the above Fig.~\ref{fig:figure4} we can observe the following:
\begin{enumerate}[(a)]
    \item For different values of $\beta_{2}$, each state variables are involved in the model \eqref{Eq1}, which converges to a steady state, indicating the local stability of the state variable with the variation of the parameter $\beta_{2}$.
    \item If $\beta_{2}$ decreases, then the number of unskilled unemployed increases, and if $\beta_{2}$ is large, then the number of unskilled unemployed almost diminishes.
    \item There is almost no effect on the employed class for the different values of $\beta_{2}$. This leads to the fact that the number of employed persons is almost independent of different values of $\beta_{2}$.
    \item If $\beta_{2}$ increases discouraged individuals increases and decreases if $\beta_{2}$ decreases.
    \item If $\beta_{2}$ increases, then the vacancies reduces and vice versa. This indicates that a large $\beta_{2}$ enhances the vacancy-filling efficiency by improving the connection between job seekers and available vacancies.
\end{enumerate}
\begin{figure}[H]
\begin{center}
\includegraphics[width=\textwidth]{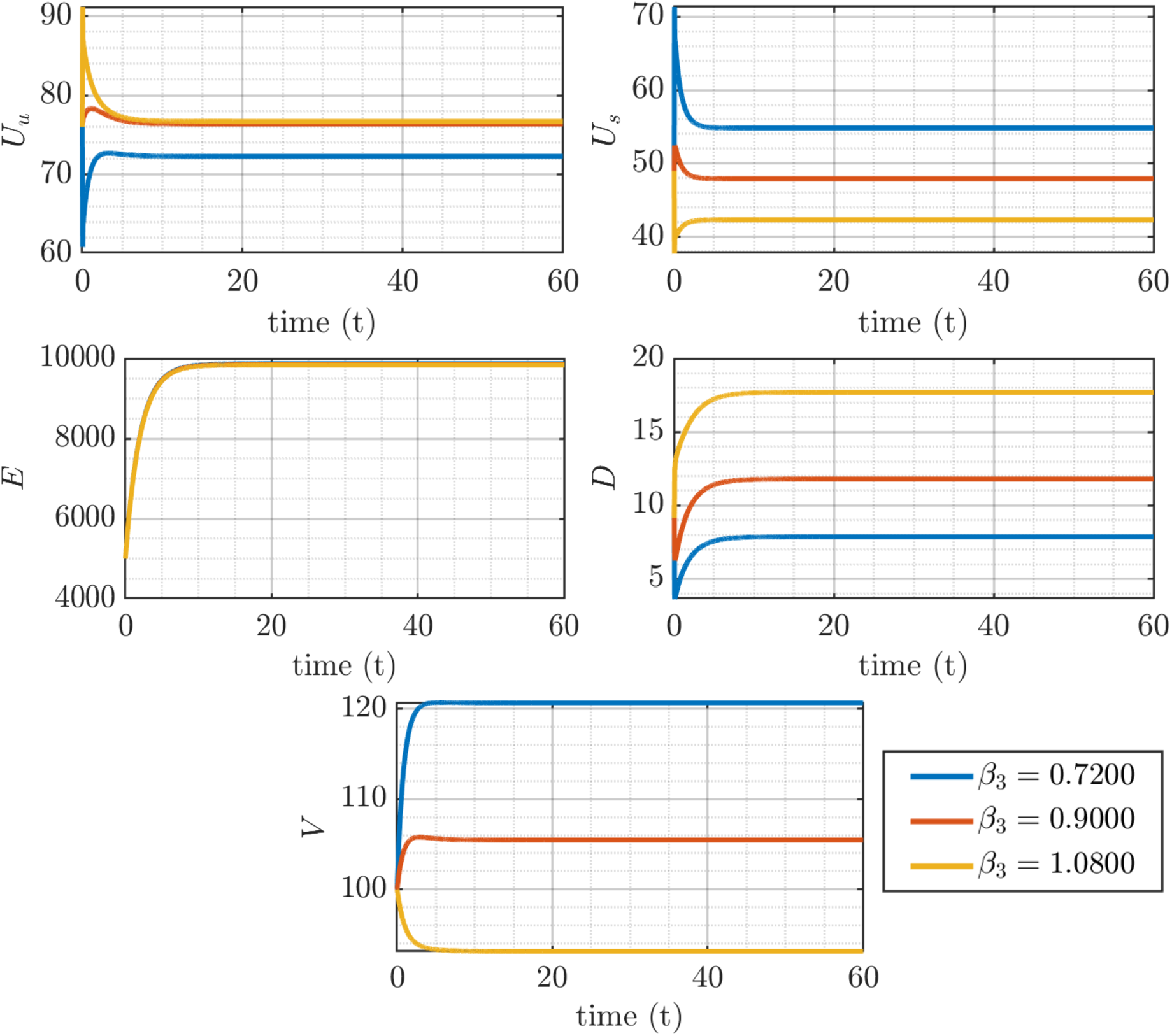}
\end{center}
\caption{Temporal dynamics of the $U_{u}(t), U_{s}(t), E(t), D(t),$ and $V(t)$ with respect to time (in days) with different values of $\beta_{3}.$}
\label{fig:figure5}
\end{figure}
From the above Fig.~\ref{fig:figure5} we can observe the following:
\begin{enumerate}[(a)]
    \item For different values of $\beta_{3}$, each state variables are involved in the model \eqref{Eq1} converges to a steady state, which indicates the local stability of the state variable with the variation of the parameter $\beta_{3}$.
    \item As $\beta_3$ increases to a higher equilibrium level of unskilled unemployed individuals $U_u(t)$, while the steady-state level of skilled unemployed individuals $U_s(t)$ decreases.
    \item The discouraged population $D(t)$ increases with $\beta_3$, whereas the number of vacancies $V(t)$ decreases significantly, indicating stronger vacancy absorption.
    \item There is almost no effect on the employed class for the different values of $\beta_{3}$. This leads to the fact that the number of employed persons is almost independent of distinct values of $\beta_{3}$.

\end{enumerate}

\begin{figure}[H]
\begin{center}
\includegraphics[width=\textwidth]{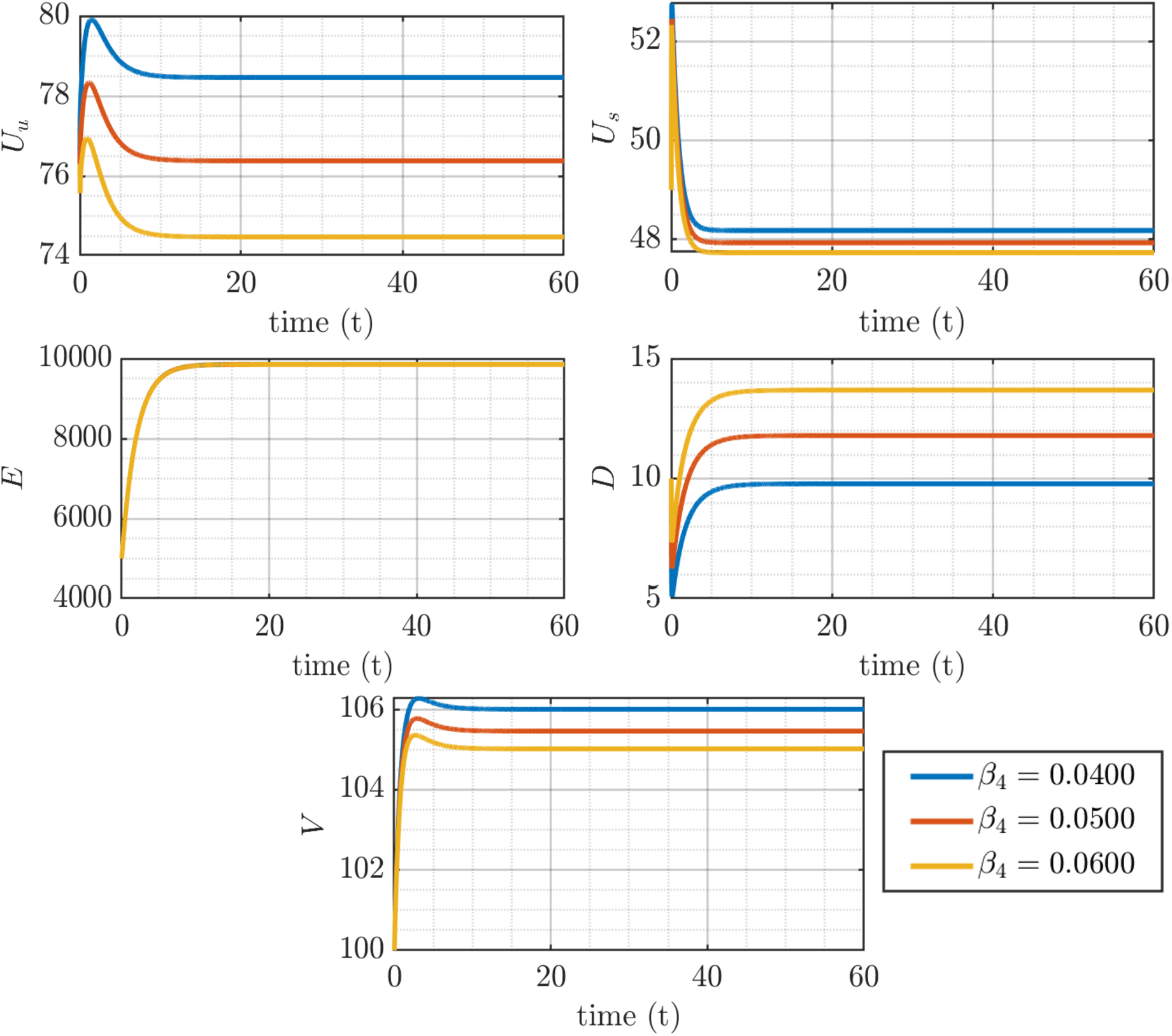}

\end{center}
\caption{Temporal dynamics of the $U_{u}(t), U_{s}(t), E(t), D(t),$ and $V(t)$ with respect to time (in days) with different values of $\beta_{4}.$}
\label{fig:figure6}
\end{figure}
From the above Fig.~\ref{fig:figure6} we can observe the following:
\begin{enumerate}[(a)]
    \item For different values of $\beta_{4}$, each state variables are involved in the model \eqref{Eq1} converges to a steady state, which shows that the local stability of the state variable with the variation of the parameter $\beta_{4}$.
    \item As $\beta_4$ increases, then equilibrium level of unskilled unemployed individuals $U_u(t)$ decreases and vice versa. Similarly, for the steady-state level of skilled unemployed individuals, $U_s(t)$ decreases when $\beta_{4}$ increases.
     \item There is almost no effect on the employed class for the different values of $\beta_{4}$. This leads to the fact that the number of employed persons is almost independent of various values of $\beta_{4}$.
    \item If the rate $\beta_{4}$ increases, then the discouraged population increases, and decreases when $\beta_{4}$ decreases. While the opposite phenomenon occurs for the vacancies, that is, if $\beta_{4}$ increases, then vacancies decrease and vice versa.

\end{enumerate}

\begin{figure}[H]
\begin{center}
\includegraphics[width=\textwidth]{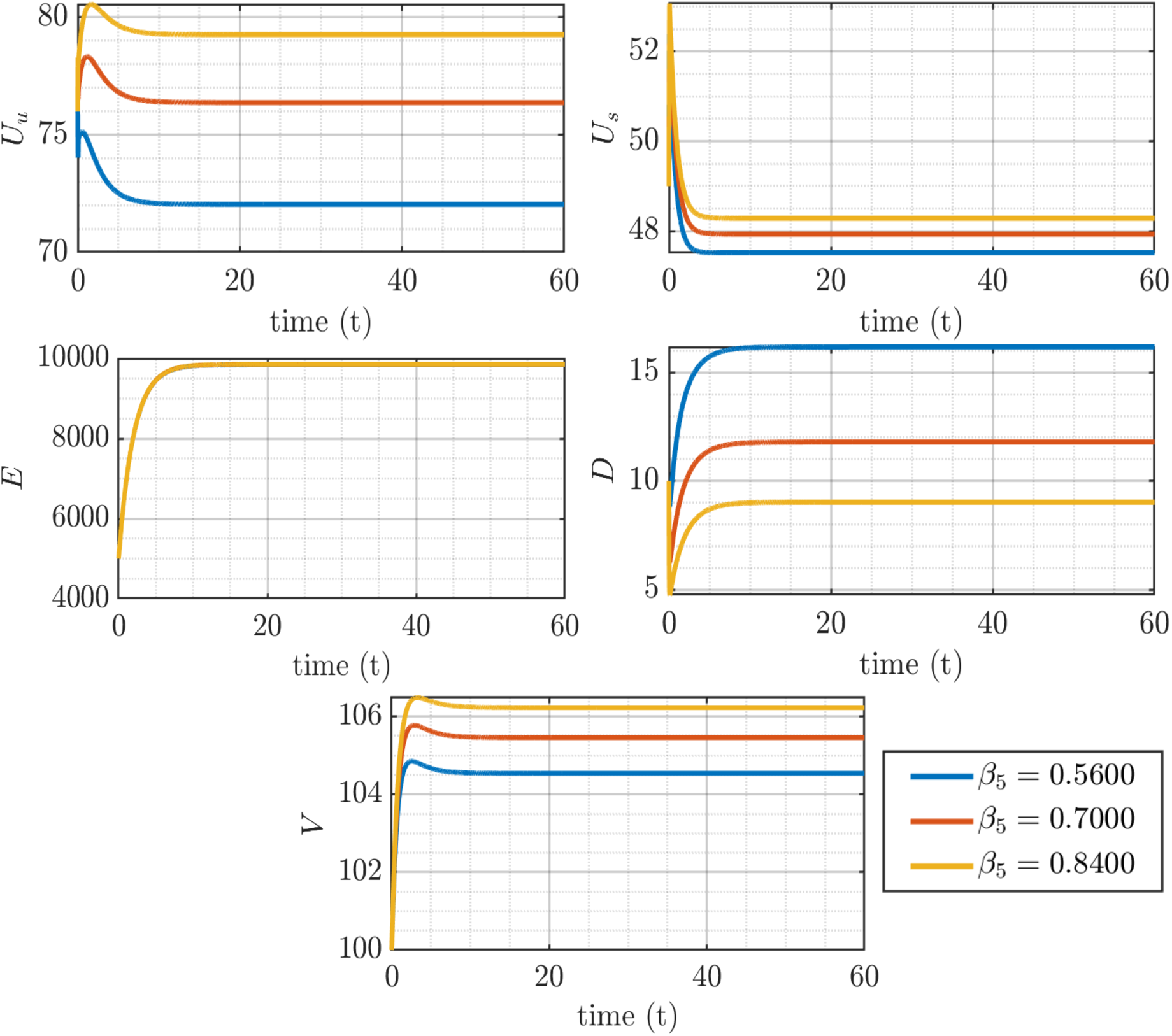}
\end{center}
\caption{Temporal dynamics of the $U_{u}(t), U_{s}(t), E(t), D(t),$ and $V(t)$ with respect to time (in days) with different values of $\beta_{5}.$}
\label{fig:figure7}
\end{figure}
From above Fig.~\ref{fig:figure7} the followings can be observed:
\begin{enumerate}[(a)]
    \item The state variables involved in the proposed model \eqref{Eq1} converge to a specific steady state for various values of $\beta_{5}$, which signifies the local stability of the state variable with the variation of the parameter $\beta_{5}$.
    \item As $\beta_{5}$ (that is, the rate at which discouraged individuals are getting jobs and moving to the employed class) increases, both unskilled and skilled unemployed persons increase and decrease when $\beta_{5}$ increases.
    \item There is almost no effect on the employed class for the different values of $\beta_{5}$. This leads to the fact that the number of employed persons is almost independent of distinct values of $\beta_{5}$.
    \item If the rate $\beta_{5}$ increases, then the vacancies increases, and decreases when $\beta_{5}$ decreases. While the opposite phenomenon occurs for the discouraged population, that is, if $\beta_{5}$ increases, then vacancies decrease and vice versa.
\end{enumerate}

\begin{figure}[H]
\begin{center}
\includegraphics[width=0.9\textwidth]{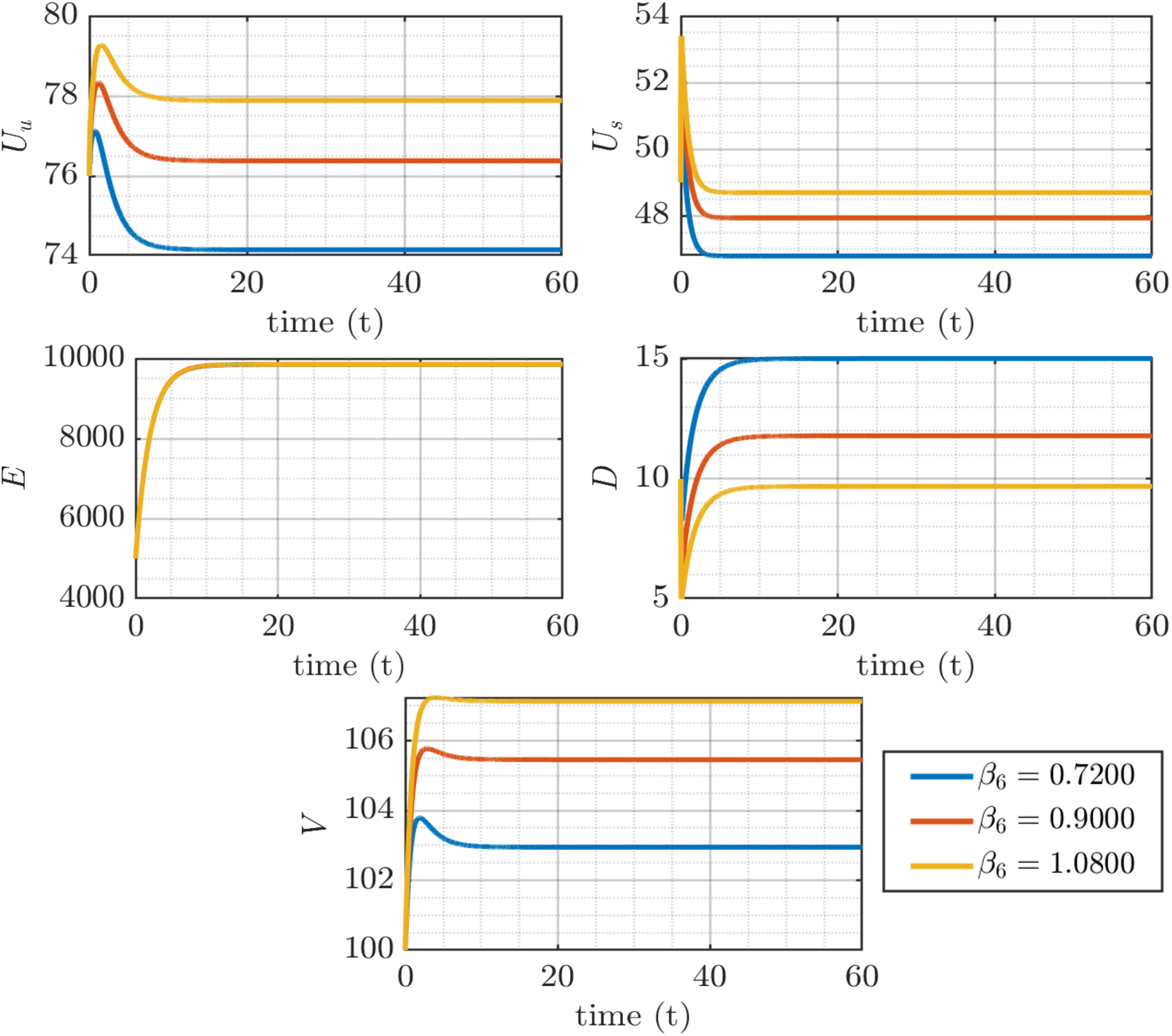}
\end{center}
\caption{Temporal dynamics of the $U_{u}(t), U_{s}(t), E(t), D(t),$ and $V(t)$ with respect to time (in days) with different values of $\beta_{6}.$}
\label{fig:figure8}
\end{figure}
One can observe the following from the above Fig.~\ref{fig:figure8}:
\begin{enumerate}[(a)]
    \item For different values of $\beta_{6}$, all state variables involved in model~\eqref{Eq1} converge to a steady state, indicating the local stability of the system with respect to variations in the parameter $\beta_{6}$.
    \item As $\beta_{6}$ (that is, the rate at which discouraged individuals are moving to the skilled unemployed class) increases, both unskilled and skilled unemployed persons increase and decrease when $\beta_{6}$ increases.
    \item There is almost no effect on the employed class for the different values of $\beta_{6}$. This leads to the fact that the number of employed persons is almost independent of various values of $\beta_{6}$.
    \item If the rate $\beta_{6}$ increases, then the vacancies increases, and decreases when $\beta_{5}$ decreases. While the opposite phenomenon occurs for the discouraged population, that is, if $\beta_{6}$ increases, then vacancies decrease and vice versa.
\end{enumerate}

\section{Conclusion}\label{S7}
\parindent=0mm\vspace{0.1in}
	In this study, a nonlinear mathematical model is considered and analyzed to determine the effect of skill development on the populations of unskilled, skilled unemployed individuals, and those who are discouraged. In the modeling process, we have considered five key variables: unskilled unemployed individuals, skilled unemployed individuals, employed individuals, discouraged individuals, and vacancies. In our proposed model we have obtained three equilibrium points $P_{0}$, $P_{1}$, and $P_{2}$. Furthermore, the stability of the critical points has been investigated by computing the eigenvalues for the equilibrium points $P_{0}$ and $P_{1}$, and by applying the Routh-Hurwitz stability criterion for the equilibrium point $P_{2}$. In addition, numerical experiments have been conducted to present solution trajectories and perform a sensitivity analysis of the dynamical variables with respect to variations in several key parameters.

    \parindent=8mm \vspace{.1in}
     The model analysis indicates that when unskilled individuals acquire skills at a higher rate through the Indian government's skill development programs, the number of discouraged individuals decreases. Moreover, as the pool of skilled unemployed individuals expands, the government responds by creating additional job vacancies specifically targeted at skilled workers. Furthermore, if a significant proportion of unskilled unemployed individuals lose hope of securing employment and transition into the discouraged worker category, then the number of skilled individuals decreases, which in turn leads to a reduction in available job vacancies. If a larger number of skilled individuals secure employment, the number of vacancies decreases, and as a consequence, discouragement among individuals increases due to the perception of reduced job opportunities. 

   \parindent=8mm \vspace{.1in}
     Moreover, when employed individuals lose their jobs or employers terminate their employment and they subsequently move into the discouraged class, the number of unskilled unemployed individuals decreases, as does the number of skilled unemployed individuals, which consequently leads to a reduction in job vacancies. If discouraged individuals obtain employment and transition into the employed class, this outcome motivates unskilled individuals to acquire skills, which subsequently leads to an increase in job vacancies. In addition, to discourage individuals from acquiring skills and entering the category of the skilled unemployed, the Government must introduce new vacancies for skilled unemployed individuals.
	\section*{Declarations}
	
	\noindent\textbf{Data availability:} Data sharing does not apply to this article as no datasets were generated or analysed during the current study.
	
	\vspace{1mm}
	\noindent\textbf{Conflict of interest:} The authors declare that they have no competing interests.
	
	\vspace{1mm}
	\noindent\textbf{Author contributions:} All authors contributed equally to this work.
	
	\vspace{1mm}
	\noindent\textbf{Acknowledgments:} The second and third authors acknowledge financial support from the University Grants Commission (UGC) under Ref. No.231620209305, 221610002999.

\bibliographystyle{unsrt}
\bibliography{Refer}
\pagebreak

\end{document}